\documentclass[11pt]{article}
\usepackage{amsmath, amssymb, amsthm, mathtools}
\usepackage[citestyle = authoryear, natbib = true]{biblatex}

\usepackage[bookmarks, bookmarksnumbered, linktocpage, urlcolor = cyan]{hyperref}
\usepackage{cleveref}
\bibliography{biblio}

\newtheorem{remark}{Remark}
\newtheorem{theorem}{Theorem}

\title{Target Weight Mechanism doesn't make delta hedge easier}

\author{Ruichao Jiang \and Long Wen}

\allowdisplaybreaks

\begin{document}
\maketitle
    \begin{abstract}
        \citet{chitra} claim that Target Weight Mechanism (TWM) in Perpetual Demand Lending Pools (PDLPs) can lower the delta of the portfolio under certain condition. We prove that their condition is self-contradictory. Furthermore, we prove an impossibility result that no TWM can lower the delta uniformly.
    \end{abstract}
    \section{Introduction} \label{sec:intro}
A PDLP is a multi-asset pool where traders can borrow  to open positions on an associated perpetual futures exchange and Liquidity Providers (LPs) to the PDLP earn fees paid by borrowers \citep{chitra}. The protocol specifies a \emph{target} weight for the pool and uses TWM to move the weight to the target: Reward/penalty on minting or redeeming pool shares that moves the weight closer/farther to the target so that \emph{arbitrageurs} can earn step in. Jupiter's JLP and Hyperliquid's HLP are two examples of PDLP.

\citet{chitra} claim that under certain condition, TWM can lower the delta (gradient of the value function with respect to price) of the whole pool, which makes delta hedging easier.

In this article, we prove that the condition for delta reduction in \citet{chitra} is self-contradictory, which automatically nullifies their proof. We remark that our refutation is purely mathematical and doesn't make use of any additional properties provided by TWM. Furthermore, we prove that a uniform reduction of delta is impossible for any TWM with very mild economical assumption.

The organization of this article is as follows. We fix the notation in \Cref{sec:background}. We prove that the condition for delta reduction in \citet{chitra} is self-contradictory in \Cref{sec:refute}. We prove our impossibility result in \Cref{sec:uniform}. We conclude in \Cref{sec:conclusion}
    \section{Background} \label{sec:background}
We fix the notation.

There are $n$ assets with strictly positive oracle prices $p\in\mathbb{R}^n_{++}$ and on-chain reserves $R\in\mathbb{R}^n_{+}$. Denote by $\ell$ the lent asset. Borrowers pay a lending fee at rate $f\in(0,1)$ hence the income to LPs is $f\cdot p$.

The \emph{mark-to-market} weight vector is defined by
\begin{equation*}
    w\coloneqq\frac{p\odot R}{p\cdot R},
\end{equation*}
where $\odot$ is the Hadamard product, i.e. $w_i=\frac{p_iR_i}{p\cdot R}$. Denote by $w^\star$ the target weight.

The TWM is as follows. The protocol dilutes or concentrates the value of the pool via a \emph{discount function} $F:\mathbb{R}^n_{++}\times\Delta\times\mathbb{R}^n_{+}\times\mathbb{R}^n\to\mathbb{R}$, where $\Delta$ is the $n$-dimensional simplex, as follows
\begin{equation} \label{eq:Vnew-bg}
    V_{\text{new}}(p)=\frac{p\cdot R}{1+F(p)}+f\,p\cdot\ell,
\end{equation}
$F$ is assumed to satisfy the following conditions \citep{chitra}, where $\delta\in\mathbb{R}^n$ denotes the reserve change tendered by the arbitrageur:
\begin{enumerate}
    \item[\textbf{(C1)}] \begin{equation} \label{eq:C1}
        F(p,w^\star,R,0)=0
    \end{equation}
    for all $p$, $w^\star$, and $R$.
    \item[\textbf{(C2)}] $F$ is concave in~$\delta$.
    \item[\textbf{(C3)}] 
    \begin{equation*}
        \arg\max_\delta F(p,w^\star,R,\delta)
    \end{equation*}
    achieves the target weight $w^\star$.
\end{enumerate}
\begin{remark}
    It would be a better notation if $w^\star$ is moved to the subscript to make it a parameter of $F$. Also, $F$ as a multiplier for the pool value, it should be a state function. But the explicit dependence on the process variable $\delta$ makes it extremely confusing. We guess \citet{chitra} might want to mean if there's no change of the reserve, the discount $F$ remains unchanged.
\end{remark}
    \section{Refutation} \label{sec:refute}
In this section, we prove that the condition of Claim 3.4 in \citet{chitra} is self-contradictory.

The condition of Claim 3.4 in \citet{chitra} is
\begin{enumerate}
    \item $\nabla_pF\geq\frac{8fR}{p\cdot R}$ (component-wise),
    \item $F\leq1$,
\end{enumerate}
where $f \in (0,1)$ is the lending fee in the PDLP model.

\begin{theorem}
\label{clm:incompatible}
    The above condition is self-contradictory.
\end{theorem}

\begin{proof}
    Choose $p_1<p_2$ (component-wise). Let $\mathcal{C}$ be the interval joining $p_1$ and $p_2$. Parametrize $\mathcal{C}$ by
    \begin{equation*}
        p(t)=p_1+(p_2-p_1)t,
    \end{equation*}
    $t\in[0,1]$.
    
    Then
    \begin{align*}
        &\int_{\mathcal{C}}\nabla_pF\cdot dp-\int_{\mathcal{C}}\frac{8fR}{p\cdot R}\cdot dp\\
        =&\sum_i\int_0^1\left(\nabla_pF-\frac{8fR}{p\cdot R}\right)_i\dot{p}_i(t)dt\\
        =&\sum_i\int_0^1\left(\nabla_pF-\frac{8fR}{p\cdot R}\right)_i(p_2-p_1)_idt\\
        \geq&0.
    \end{align*}
    On the other hand,
    \begin{equation*}
         \int_{\mathcal{C}}\nabla_pF\cdot dp=F(p_2)-F(p_1),
    \end{equation*}
    and
    \begin{equation*}
        \int_{\mathcal{C}}\frac{8fR}{p\cdot R}\cdot dp=8f\ln\frac{p_2\cdot R}{p_1\cdot R}.
    \end{equation*}
    Hence,
    \begin{equation*}
        F(p_2)-F(p_1)\geq8f\ln\frac{p_2\cdot R}{p_1\cdot R}.
    \end{equation*}
    Since $F\leq1$,
    \begin{equation*}
        1-F(p_1)\geq8f\ln\frac{p_2^TR}{p_1^TR}.
    \end{equation*}
    Let $||p_2||\to\infty$, the RHS of the above would be unbounded whereas the LHS would remain bounded, a contradiction.
\end{proof}
\begin{remark}
    The above proof didn't use any economical axioms of the discount function $F$. The refutation is purely a mathematical consequence.
\end{remark}
    \section{No uniform delta shrinkage} \label{sec:uniform}
We prove that under a very mild economical assumption, no TWM can achieve a uniform delta shrinkage. The assumption is as follows.
\begin{equation}\label{eq:Vprime-homog1}
    V_{\text{new}}(\lambda p)=\lambda V_{\text{new}}(p)
\end{equation}
for all $\lambda>0$.

\Cref{eq:Vprime-homog1} says that if the value of the pool under TWM is $1$ dollar, it's also $100$ cents.

Then \Cref{eq:Vnew-bg} implies that
\begin{equation*}
    F(\lambda p)=F(p)
\end{equation*}
for all $\lambda>0$, i.e. $F$ is homogeneous of degree zero in $p$. By Euler's theorem,
\begin{equation} \label{eqn:euler}
    \nabla_p F\cdot p=0
\end{equation}
for all $p\in\mathbb{R}^n_{++}$.
\begin{theorem}\label{thm:uniform-impossible}
    If the value function $V_\text{new}$ under the TWM respects change of numeraire, then there doesn't exist $0<c<1$ s.t.
    \begin{equation*}
        \nabla_p V_{\text{new}}\leq c\nabla_p V_{\text{old}},
    \end{equation*}
    where $\nabla_p V_{\text{old}}=R$.
\end{theorem}
\begin{proof}
    Recall from \citet{chitra} that
    \begin{equation*}\label{eq:grad-Vnew}
    \nabla_p V_{\text{new}}
    = \frac{R}{1+F} - \frac{S}{(1+F)^2}\,\nabla_p F + f\ell.
    \end{equation*}
    Taking the inner product of the above and $p$ and using \Cref{eqn:euler},
    \begin{equation*}
        \nabla_pV_{\text{new}}\cdot p = \frac{p\cdot R}{1+F}+\,f\,p\cdot\ell.
    \end{equation*}
    Suppose for contradiction that such $c$ existed. It would follow that
    \begin{equation}\label{eq:ineq-SF}
        \frac{p\cdot R}{1+F}+fp\cdot\ell\leq cp\cdot R.
    \end{equation}
    Since $fp\cdot\ell\geq0$,
    \begin{equation*}
        \frac{p\cdot R}{1+F}\leq cp\cdot R-fp\cdot\ell\leq cp\cdot R.
    \end{equation*}
    Hence,
    \begin{equation*}
        1+F\geq\frac{1}{c}>1\implies F>0,
    \end{equation*}
    contradicting \Cref{eq:C1}.
\end{proof}
\begin{remark}
    To reach the contradiction, we used the strongest reading of \Cref{eq:C1} from \citet{chitra}. If we abondon \Cref{eq:C1}, we can still reach the desired contradiction if we assume in addition that $F=0$ for some $p$ and $R$. Note that this is a very mild economical assumption. If $F>0$ always, it means the value of an LP's position is always diluted the moment he enters the pool, which is absurd.
\end{remark}
    \section{Conclusion} \label{sec:conclusion}
In this article, we refuted the claim in \citet{chitra} that the TWM reduces the volatility of the pool, which they used to explain the popularity of protocols such as HLP and JLP. We further proved that no uniform reduction in volatility (delta) is possible under very mild economical assumptions. One further work is to waive the strong requirement that the reduction be uniform in $p$ and study in what price range does the reduction of volatility become possible.
    \printbibliography
\end{document}